\documentclass[reqno,11pt]{amsart}

\usepackage{amsmath}
\usepackage{amsthm}
\usepackage{graphicx}
\usepackage{amssymb}
\usepackage{enumitem}

\usepackage[dvipsnames]{xcolor}

\theoremstyle{definition}
\newtheorem{Def}{Definition}
\newtheorem{definition}[Def]{Definition}

\newtheorem{Lemma}[Def]{Lemma}
\newtheorem{Remark}[Def]{Remark}

\newtheorem{example}[Def]{Example}

\newcommand{\Decomps}{\mathbb{D}} 

\newcommand{\Thanks}{\vspace*{.5em} \noindent \thanks}

\newcommand{\R}{\mathbb{R}}

\DeclareMathOperator{\Tr}{Tr}

\DeclareMathOperator{\norm}{| \hspace*{-0.1em}| }

\DeclareFontFamily{OT1}{rsfso}{}
\DeclareFontShape{OT1}{rsfso}{m}{n}{ <-7> rsfso5 <7-10> rsfso7 <10-> rsfso10}{}
\DeclareMathAlphabet{\mycal}{OT1}{rsfso}{m}{n}

\newcommand{\bei}{\begin{itemize}[itemsep=.2em]}
\newcommand{\eni}{\end{itemize}}

\usepackage{tikz,xypic}
\usetikzlibrary{decorations.pathreplacing,decorations.markings,arrows.meta,backgrounds,shapes}
\usetikzlibrary{circuits.ee.IEC}
\pgfdeclarelayer{edgelayer}
\pgfdeclarelayer{nodelayer}
\pgfsetlayers{background,edgelayer,nodelayer,main}
\tikzstyle{none}=[inner sep=0mm]
\tikzstyle{every loop}=[]
\tikzstyle{mark coordinate}=[inner sep=0pt,outer sep=0pt,minimum size=3pt,fill=black,circle]

\tikzset{arrow/.style={decoration={
    markings,
    mark=at position #1 with \arrow{>[length=2pt, width=3pt]}},
    postaction=decorate},
    reverse arrow/.style={decoration={
    markings,
    mark=at position #1 with {{\arrow{<[length=2pt, width=3pt]}}}},
    postaction=decorate}
}

\tikzstyle{upground}=[circuit ee IEC,thick,ground,rotate=90,scale=1.5]

\newcommand{\sys}[1]{\ensuremath{\mathbf{#1}}}
\newcommand{\sysS}{\sys{S}}
\newcommand{\caus}{\mathsf{caus}}
\newcommand{\eff}{\mathsf{eff}}

\newcommand{\TZdir}{T^{\overset{\rightarrow}{Z}}}

\newcommand{\exrep}{\ensuremath{\overline{R^{(x)}}}}
\newcommand{\eRerep}{\ensuremath{\overline{R^{\eff}}}}
\newcommand{\eRcrep}{\ensuremath{\overline{R^{\caus}}}}
\newcommand{\gxrep}{\ensuremath{R^{(x)}}}
\newcommand{\gerep}{\ensuremath{R^\eff}}
\newcommand{\gcrep}{\ensuremath{R^\caus}}
\newcommand{\uxrep}{\ensuremath{R^{(x)}_\mathsf{u}}}

\newcommand{\ccaus}{P^\mathsf{c}}
\newcommand{\ceff}{P^\mathsf{e}}
\newcommand{\concept}[1]{\mathbb{C}(#1)}

\newcommand{\beq}{\begin{equation}}
\newcommand{\eeq}{\end{equation}}
\newcommand{\Proof}{\begin{proof}}
\newcommand{\QED}{\end{proof} \noindent}

\usepackage{amsmath,amssymb,stmaryrd,xspace}
\usepackage{amsthm}
\usepackage{tikz-cd}
\usepackage{multirow}
\usepackage{mathtools}
\usepackage{url}
\usepackage{relsize}
\usepackage{bm} 
\newcounter{counter}

\newcommand{\Expcat}{\cat{Exp}}

\newcommand{\PExp}{\mathbb{PE}}
\newcommand{\Exp}{\mathbb{E}}

\newcommand{\Sys}{\cat{Sys}}

\newcommand{\BasicIIT}{\mathsf{IIT}}

\newcommand{\QShape}{\mathbb{Q}}
\newcommand{\cut}[2]{{#1}^{#2}}

\newcommand{\marg}[2]{{#1}|_{#2}} 

\newcommand{\notetoself}[1]{}
\newcommand{\omitfornow}[1]{}
\newcommand{\omitthis}[1]{}

\newcommand{\Class}{\cat{Class}}
\newcommand{\Classm}{\cat{Class}_\mathsf{m}} 

\newcommand{\FCStar}{\cat{CStar}}

\newcommand{\Quant}[1]{\cat{Quant}_{#1}}

\newcommand{\Split}[1]{\ensuremath{\mathrm{Split}}\xspace}

\newcommand{\cat}[1]{\ensuremath{\mathbf{#1}}\xspace}

\newcommand{\catC}{\cat{C}}

\newcommand{\Rplus}{\mathbb{R}^+} 

\newcommand{\id}[1]{\ensuremath{\mathrm{id}_{#1}}}
\newcommand{\Sub}{\mathrm{Sub}}

\newcommand{\discard}[1]{\ensuremath{\tinygroundnew_{#1}}}

\newcommand{\discardflip}[1]{\ensuremath{\tinygroundflipnew_{#1}}}

\DeclareFontFamily{U}{mathux}{\hyphenchar\font45}
\DeclareFontShape{U}{mathux}{m}{n}{
      <5> <6> <7> <8> <9> <10>
      <10.95> <12> <14.4> <17.28> <20.74> <24.88>
      mathux10
      }{}
\DeclareSymbolFont{mathux}{U}{mathux}{m}{n}

\DeclareMathSymbol{\bigovee}{1}{mathux}{"8F}

\DeclareMathSymbol{\bigperp}{1}{mathux}{"4E}

\DeclareFontFamily{U}{mathb}{\hyphenchar\font45}
\DeclareFontShape{U}{mathb}{m}{n}{
      <5> <6> <7> <8> <9> <10> gen * mathb
      <10.95> mathb10 <12> <14.4> <17.28> <20.74> <24.88> mathb12
      }{}
\DeclareSymbolFont{mathb}{U}{mathb}{m}{n}
\DeclareFontSubstitution{U}{mathb}{m}{n}
\DeclareMathSymbol{\mylgroup}{\mathbin}{mathb}{'160}
\DeclareMathSymbol{\myrgroup}{\mathbin}{mathb}{'161}

\newcommand{\Stc}{\St_\mathsf{c}}
\newcommand{\St}{\text{St}}

\newcommand{\hilbH}{\mathcal{H}} 
\newcommand{\hilbK}{\mathcal{K}} 

\usepackage{tikz,xypic}

\usetikzlibrary{decorations.pathreplacing,decorations.markings,arrows.meta,backgrounds}
\pgfdeclarelayer{edgelayer}
\pgfdeclarelayer{nodelayer}
\pgfsetlayers{background,edgelayer,nodelayer,main}

\tikzstyle{cdot}=[circle, draw=black, fill=black!25, inner sep=.4ex] 
\tikzstyle{bigdot}=[dot, inner sep=0pt]
\tikzstyle{whitedot}=[circle, draw=black, fill=white, inner sep=.4ex]
\tikzstyle{greydot}=[circle, draw=black, fill=black!25, inner sep=.4ex] 
\tikzstyle{blackdot}=[circle, draw=black, fill=black, inner sep=.4ex]
\tikzset{arrow/.style={decoration={
    markings,
    mark=at position #1 with \arrow{>[length=2pt, width=3pt]}},
    postaction=decorate},
    reverse arrow/.style={decoration={
    markings,
    mark=at position #1 with {{\arrow{<[length=2pt, width=3pt]}}}},
    postaction=decorate}
}

\newcommand{\tinyswap}{
\smash{\raisebox{-1pt}{\hspace{-2pt}\ensuremath{
   \begin{pic}[scale=.25]
  \draw (0,-.5) to[out=80,in=-100] (1,.5);
  \draw (1,-.5) to[out=100,in=-80] (0,.5);
\end{pic}}}}
}
\newenvironment{pic}[1][] {\begin{aligned}\begin{tikzpicture}[scale=2.0, font=\tiny,#1]}{\end{tikzpicture}\end{aligned}} 

\newif\ifvflip\pgfkeys{/tikz/vflip/.is if=vflip}
\newif\ifhflip\pgfkeys{/tikz/hflip/.is if=hflip}
\newif\ifhvflip\pgfkeys{/tikz/hvflip/.is if=hvflip}

\newenvironment{picc}[1][]
{\begin{aligned}\begin{tikzpicture}[font=\tiny,#1]}
{\end{tikzpicture}\end{aligned}}

\newlength\minimummorphismwidth
\newlength\stateheight
\newlength\minimumstatewidth
\newlength\connectheight
\tikzset{colour/.initial=white}

\tikzstyle{pure}=[line width=.7pt]

\makeatletter

\pgfdeclareshape{groundd}
{
    \savedanchor\centerpoint
    {
        \pgf@x=0pt
        \pgf@y=0pt
    }
    \anchor{center}{\centerpoint}
    \anchorborder{\centerpoint}

    \anchor{north}
    {
        \pgf@x=0pt
        \pgf@y=0.16\stateheight
    }
    \anchor{south}
    {
        \pgf@x=0pt
        \pgf@y=0pt
    }
    \saveddimen\overallwidth
    {
        \pgfkeysgetvalue{/pgf/minimum width}{\minwidth}
        \pgf@x=\minimumstatewidth
        \ifdim\pgf@x<\minwidth
            \pgf@x=\minwidth
        \fi
    }
    \backgroundpath
    {
        \begin{pgfonlayer}{main} 
        \pgfsetstrokecolor{black}
        \pgfsetlinewidth{1.25pt}
        \ifhflip
            \pgftransformyscale{-1}
        \fi
        \pgftransformscale{0.5}
        \pgfpathmoveto{\pgfpoint{-0.5*\overallwidth}{0pt}}
        \pgfpathlineto{\pgfpoint{0.5*\overallwidth}{0pt}}
        \pgfpathmoveto{\pgfpoint{-0.33*\overallwidth}{0.33*\stateheight}}
        \pgfpathlineto{\pgfpoint{0.33*\overallwidth}{0.33*\stateheight}}
        \pgfpathmoveto{\pgfpoint{-0.16*\overallwidth}{0.66*\stateheight}}
        \pgfpathlineto{\pgfpoint{0.16*\overallwidth}{0.66*\stateheight}}
        \pgfpathmoveto{\pgfpoint{-0.02*\overallwidth}{\stateheight}}
        \pgfpathlineto{\pgfpoint{0.02*\overallwidth}{\stateheight}}
        \pgfusepath{stroke}
        \end{pgfonlayer}
    }
}

\usepackage{tikz,xypic}
\usetikzlibrary{decorations.pathreplacing,decorations.markings,arrows.meta,backgrounds,shapes}
\usetikzlibrary{circuits.ee.IEC}
\pgfdeclarelayer{edgelayer}
\pgfdeclarelayer{nodelayer}
\pgfsetlayers{background,edgelayer,nodelayer,main}
\tikzstyle{none}=[inner sep=0mm]
\tikzstyle{every loop}=[]
\tikzstyle{mark coordinate}=[inner sep=0pt,outer sep=0pt,minimum size=3pt,fill=black,circle]

\tikzset{arrow/.style={decoration={
    markings,
    mark=at position #1 with \arrow{>[length=2pt, width=3pt]}},
    postaction=decorate},
    reverse arrow/.style={decoration={
    markings,
    mark=at position #1 with {{\arrow{<[length=2pt, width=3pt]}}}},
    postaction=decorate}
}

\tikzstyle{upground}=[circuit ee IEC,thick,ground,rotate=90,scale=1.5]
\tikzstyle{upgroundwhite}=[circuit ee IEC,thick,ground,rotate=90,scale=1.5, fill=white]
\tikzstyle{downground}=[circuit ee IEC,thick,ground,rotate=-90,scale=1.5]
\tikzstyle{downgroundnorm}=[circuit ee IEC,thick,ground,rotate=-90,scale=1.5, fill=white]

\newcommand{\mapminh}{5mm} 
\newcommand{\stateminh}{5mm}
\newcommand{\maplw}{0.7pt} 
\newcommand{\stateshift}{-0.2pt}
\newcommand{\effectshift}{-0.2pt}

\tikzstyle{box}=[map]
\tikzstyle{medium box}=[medium map]
\tikzstyle{dot}=[inner sep=0mm,minimum width=2mm,minimum height=2mm,draw,shape=circle]  
\tikzstyle{black dot}=[dot,fill=black]
\tikzstyle{white dot}=[dot,fill=white,,text depth=-0.2mm]
\tikzstyle{grey dot}=[dot,fill=black!25] 

\tikzstyle{corner1}=[box,fill=white, font=\footnotesize] %
\tikzstyle{corner2}=[dot,fill=white, font=\footnotesize] %
\tikzstyle{corner3}=[dot,fill=black!25, font=\footnotesize] %
\tikzstyle{corner4}=[dot,fill=black, font=\footnotesize] %

\tikzstyle{scalar}=[circle,draw,inner sep=2pt, line width=\maplw] 

\usetikzlibrary{shapes.misc, positioning}

\tikzset{stateshape/.style={append after command={
   \pgfextra
        \draw[sharp corners, fill=white, line width = \maplw]%
    (\tikzlastnode.west)%
    [rounded corners=0pt] |- (\tikzlastnode.north)%
    [rounded corners=0pt] -| (\tikzlastnode.east)%
    [rounded corners=5pt] |- (\tikzlastnode.south)%
    [rounded corners=5pt] -| (\tikzlastnode.west);
   \endpgfextra}}}

\tikzset{effectshape/.style={append after command={
   \pgfextra
        \draw[sharp corners, fill=white, line width = \maplw]%
    (\tikzlastnode.west)%
    [rounded corners=0pt] |- (\tikzlastnode.south)%
    [rounded corners=0pt] -| (\tikzlastnode.east)%
    [rounded corners=5pt] |- (\tikzlastnode.north)%
    [rounded corners=5pt] -| (\tikzlastnode.west);
   \endpgfextra}}}

 \tikzstyle{map}=[draw,shape=rectangle, inner sep=2pt,minimum height=\mapminh, minimum width=7mm,fill=white]

\tikzstyle{point}=[stateshape,inner sep=2pt, minimum width=6mm, minimum height=\stateminh, yshift=\stateshift]
\tikzstyle{copoint}=[effectshape,inner sep=.2pt, minimum width=6mm, minimum height=\stateminh, yshift=-\effectshift]

\tikzstyle{wide point}=[point, minimum width=12mm]
\tikzstyle{wide copoint}=[copoint, minimum width=12mm]

\tikzstyle{decomp}=[fill=white,draw,shape=isosceles triangle,shape border rotate=-90,isosceles triangle stretches=true,inner sep=0pt,minimum width=0.75cm,minimum height=4mm,yshift=-0.0mm]

\tikzstyle{decompwide}=[fill=white,draw,shape=isosceles triangle,shape border rotate=-90,isosceles triangle stretches=true,inner sep=0pt,minimum width=1.5cm,minimum height=4mm,yshift=-0.0mm]

\tikzstyle{decompflip}=[fill=white,draw,shape=isosceles triangle,shape border rotate=90,isosceles triangle stretches=true,inner sep=0pt,minimum width=0.75cm,minimum height=4mm,yshift=-0.0mm]

\tikzstyle{decompwideflip}=[fill=white,draw,shape=isosceles triangle,shape border rotate=90,isosceles triangle stretches=true,inner sep=0pt,minimum width=1.5cm,minimum height=4mm,yshift=-0.0mm]

 \tikzstyle{map}=[draw,shape=rectangle, inner sep=2pt,minimum height=\mapminh, minimum width=7mm,fill=white, line width = \maplw]

\tikzstyle{medium map} = [map, minimum width = 12mm] 
\tikzstyle{semilarge map} = [map, minimum width = 15mm] 
\tikzstyle{large map} = [map, minimum width = 18mm] 

\tikzstyle{kpoint} =[point]
\tikzstyle{kpointadj} =[copoint]
\tikzstyle{kpointconj}=[dagpointconj] 

\makeatletter
\newcommand{\boxshape}[3]{%
\pgfdeclareshape{#1}{
\inheritsavedanchors[from=rectangle] 
\inheritanchorborder[from=rectangle]
\inheritanchor[from=rectangle]{center}
\inheritanchor[from=rectangle]{north}
\inheritanchor[from=rectangle]{south}
\inheritanchor[from=rectangle]{west}
\inheritanchor[from=rectangle]{east}
\backgroundpath{
\southwest \pgf@xa=\pgf@x \pgf@ya=\pgf@y
\northeast \pgf@xb=\pgf@x \pgf@yb=\pgf@y

\@tempdima=#2
\@tempdimb=#3

\pgfpathmoveto{\pgfpoint{\pgf@xa - 5pt + \@tempdima}{\pgf@ya}}
\pgfpathlineto{\pgfpoint{\pgf@xa - 5pt - \@tempdima}{\pgf@yb}}
\pgfpathlineto{\pgfpoint{\pgf@xb + 5pt + \@tempdimb}{\pgf@yb}}
\pgfpathlineto{\pgfpoint{\pgf@xb + 5pt - \@tempdimb}{\pgf@ya}}
\pgfpathlineto{\pgfpoint{\pgf@xa - 5pt + \@tempdima}{\pgf@ya}}
\pgfpathclose
}
}}

\boxshape{NEbox}{0pt}{3pt} 
\boxshape{SEbox}{0pt}{-3pt}
\boxshape{NWbox}{3pt}{0pt}
\boxshape{SWbox}{-3pt}{0pt}
\boxshape{rec-box}{0pt}{0pt}
\makeatother

\tikzstyle{cloud}=[shape=cloud,draw,minimum width=1.5cm,minimum height=1.5cm]

\tikzstyle{dagmap}=[draw,shape=NEbox,inner sep=2pt,minimum height=\mapminh,fill=white, line width = \maplw] %
\tikzstyle{dashedmap}=[draw,dashed,shape=NEbox,inner sep=2pt,minimum height=\mapminh,fill=white, line width = \maplw]
\tikzstyle{mapdag}=[draw,shape=SEbox,inner sep=2pt,minimum height=\mapminh,fill=white, line width = \maplw]
\tikzstyle{mapadj}=[draw,shape=SEbox,inner sep=2pt,minimum height=\mapminh,fill=white, line width = \maplw]
\tikzstyle{maptrans}=[draw,shape=SWbox,inner sep=2pt,minimum height=\mapminh,fill=white, line width = \maplw]
\tikzstyle{mapconj}=[draw,shape=NWbox,inner sep=2pt,minimum height=\mapminh,fill=white, line width = \maplw]

\tikzstyle{medium dagmap}=[draw,shape=NEbox,inner sep=2pt,minimum height=\mapminh,fill=white,minimum width=7mm, line width = \maplw]
\tikzstyle{semilarge dagmap}=[draw,shape=NEbox,inner sep=2pt,minimum height=\mapminh,fill=white,minimum width=9.5mm, line width = \maplw]
\tikzstyle{large dagmap}=[draw,shape=NEbox,inner sep=2pt,minimum height=\mapminh,fill=white,minimum width=12mm, line width = \maplw]

\makeatletter

\pgfdeclareshape{cornerpoint}{
\inheritsavedanchors[from=rectangle] 
\inheritanchorborder[from=rectangle]
\inheritanchor[from=rectangle]{center}
\inheritanchor[from=rectangle]{north}
\inheritanchor[from=rectangle]{south}
\inheritanchor[from=rectangle]{west}
\inheritanchor[from=rectangle]{east}
\backgroundpath{
\southwest \pgf@xa=\pgf@x \pgf@ya=\pgf@y
\northeast \pgf@xb=\pgf@x \pgf@yb=\pgf@y

\pgfmathsetmacro{\pgf@shorten@left}{\pgfkeysvalueof{/tikz/shorten left}}
\pgfmathsetmacro{\pgf@shorten@right}{\pgfkeysvalueof{/tikz/shorten right}}

\pgfpathmoveto{\pgfpoint{0.5 * (\pgf@xa + \pgf@xb)}{\pgf@ya - 5pt}}
\pgfpathlineto{\pgfpoint{\pgf@xa - 8pt + \pgf@shorten@left}{\pgf@yb - 1.5 * \pgf@shorten@left}}
\pgfpathlineto{\pgfpoint{\pgf@xa - 8pt + \pgf@shorten@left}{\pgf@yb}}
\pgfpathlineto{\pgfpoint{\pgf@xb + 8pt - \pgf@shorten@right}{\pgf@yb}}
\pgfpathlineto{\pgfpoint{\pgf@xb + 8pt - \pgf@shorten@right}{\pgf@yb - 1.5 * \pgf@shorten@right}}
\pgfpathclose
}
}

\pgfdeclareshape{cornercopoint}{
\inheritsavedanchors[from=rectangle] 
\inheritanchorborder[from=rectangle]
\inheritanchor[from=rectangle]{center}
\inheritanchor[from=rectangle]{north}
\inheritanchor[from=rectangle]{south}
\inheritanchor[from=rectangle]{west}
\inheritanchor[from=rectangle]{east}
\backgroundpath{
\southwest \pgf@xa=\pgf@x \pgf@ya=\pgf@y
\northeast \pgf@xb=\pgf@x \pgf@yb=\pgf@y

\pgfmathsetmacro{\pgf@shorten@left}{\pgfkeysvalueof{/tikz/shorten left}}
\pgfmathsetmacro{\pgf@shorten@right}{\pgfkeysvalueof{/tikz/shorten right}}

\pgfpathmoveto{\pgfpoint{0.5 * (\pgf@xa + \pgf@xb)}{\pgf@yb + 5pt}}
\pgfpathlineto{\pgfpoint{\pgf@xa - 8pt + \pgf@shorten@left}{\pgf@ya + 1.5 * \pgf@shorten@left}}
\pgfpathlineto{\pgfpoint{\pgf@xa - 8pt + \pgf@shorten@left}{\pgf@ya}}
\pgfpathlineto{\pgfpoint{\pgf@xb + 8pt - \pgf@shorten@right}{\pgf@ya}}
\pgfpathlineto{\pgfpoint{\pgf@xb + 8pt - \pgf@shorten@right}{\pgf@ya + 1.5 * \pgf@shorten@right}}
\pgfpathclose
}
}

\makeatother

\pgfkeyssetvalue{/tikz/shorten left}{0pt}
\pgfkeyssetvalue{/tikz/shorten right}{0pt}

\tikzstyle{dagpoint common}=[draw,fill=white,inner sep=1pt, line width = \maplw, minimum height = 4mm, yshift=1.2pt] 
\tikzstyle{dagpoint sc}=[shape=cornerpoint,dagpoint common]
\tikzstyle{dagpoint adjoint sc}=[shape=cornercopoint,dagpoint common]
\tikzstyle{dagpoint}=[shape=cornerpoint,shorten left=4pt,dagpoint common]
\tikzstyle{dagpointadj}=[shape=cornercopoint,shorten left=5pt,dagpoint common]
\tikzstyle{dagpointconj}=[shape=cornerpoint,shorten right=5pt,dagpoint common]
\tikzstyle{dagpointtrans}=[shape=cornercopoint,shorten right=5pt,dagpoint common]
\tikzstyle{dagpointsymm}=[shape=cornerpoint,shorten left=5pt,shorten right=5pt,dagpoint common]

\tikzstyle{widedagpoint}=[dagpoint, minimum width=1 cm, inner sep=2pt]
\tikzstyle{widedagpointadj}=[dagpointadj, minimum width=1 cm, inner sep=2pt]

\tikzstyle{every picture}=[baseline=-0.25em,scale=0.5]
\tikzstyle{label}=[font=\footnotesize,text height=1ex, text depth=0.15ex]

\usetikzlibrary[shapes]

\tikzset{
sidetriangle/.style = {regular polygon, regular polygon sides = 3, aspect = 1, shape border rotate = 90, draw, inner sep = 0, minimum width = 1.2cm}
}

\tikzset{
isoc/.style = {shape=isosceles triangle, shape border rotate = 180, isosceles triangle stretches = true, minimum width = 1.2cm, minimum height= 1.5cm, inner sep = 0.3}}

\tikzset{
coarse/.style = {shape = circle, fill = white, draw, inner sep = 0, minimum width =0.125cm}
}
\tikzset{
coarsesymbol/.style = {shape = circle, fill = white, inner sep = -0.7, minimum width = 0.125cm}
}

\tikzstyle{sidetriangle2}=[sidetriangle, minimum width = 2cm, fill=white]
\tikzstyle{sideisocsmall}]=[style=isoc, minimum width = 1cm, minimum height = 0.8cm, draw, fill=white, font=\Large]
\tikzstyle{sideisoc}]=[style=isoc, minimum width = 2cm, draw, fill=white, font=\Large]
\tikzstyle{sideisocmid}]=[style=isoc, minimum width = 2.5cm, draw, fill=white, font=\Large]
\tikzstyle{sideisocmedium}]=[style=isoc, minimum width = 3cm, draw, fill=white, font=\Large]

\newcommand{\tinygroundnew}{
\smash{
{\hspace{-3pt}
\ensuremath{
\begin{picc}[scale=1.0] 
    \node[upground, xscale=0.8, yscale=0.7] (1) at (0,0.16) {};
    \draw (0,0.03) to (0,-0.25);
\end{picc}
}\hspace{-1pt}}}}

\newcommand{\tinygroundflipnew}{
\smash{
{\hspace{-3pt}
\ensuremath{
\begin{picc}[yscale=-1.0] 
    \node[upground, xscale=0.8, yscale=-0.7] (1) at (0,0.10) {};
    \draw (0,-0.03) to (0,-0.31);
\end{picc}
}\hspace{-1pt}}}}

\usepackage{hyperref}

\title{Integrated Information in Process Theories}

\author{Sean Tull}
\email{sean.tull@cambridgequantum.com}
\address{Cambridge Quantum Computing}

\author{Johannes Kleiner}
\email{johannes.kleiner@itp.uni-hannover.de }
\address{Munich Center for Mathematical Philosophy, University of Munich}

\begin{document}
\date{}
\begin{abstract}
We demonstrate how the key notions of Tononi et al.'s Integrated Information Theory (IIT) can be studied within the simple graphical language of process theories, i.e.~symmetric monoidal categories. This allows IIT to be generalised to a broad range of physical theories, including as a special case the Quantum IIT of Zanardi, Tomka and Venuti. 
\end{abstract}

\maketitle


In recent years, a toolkit for the study of integrated causal behaviours has been developed by Giulio Tononi and collaborators under the name of \emph{Integrated Information Theory (IIT)} \cite{tononi2004information,oizumi2014phenomenology}. Primarily proposed as a scientific theory of consciousness, the theory is based on the idea that consciousness originates from integrated, or `holistic', internal dynamics in the brain. More broadly, the methods of IIT have been applied to study integrated behaviour in simple information processing systems, including autonomy \cite{marshall2017causal}, causation \cite{albantakis2017caused}, and in the study of state differentiation \cite{marshall2016integrated}. 

 While the principles behind IIT appear to be quite general, it is typically only applied to simple, finite classical physical systems (often described as graphs of interacting `elements'). In the related article~\cite{GeneralisedIITs}, the present authors have shown that the core algorithm of IIT can be significantly extended, allowing one to formally define \emph{generalised IITs} based on very broad notions of physical systems.

 In this article we show how the key concepts of IIT, including those of \emph{systems}, \emph{integration} and \emph{causation}, can be studied naturally in the language of physical \emph{process theories}, which are mathematically described as \emph{symmetric monoidal categories}. Process theories come with an intuitive but rigorous graphical calculus \cite{selinger2011survey} which allows us to present many aspects of IIT pictorially.

 In particular, we show how to define a generalised IIT starting from any suitable process theory, allowing us to extend IIT to new physical settings. Choosing the theory of classical probabilistic processes essentially yields IIT 3.0 in the sense of~\cite{oizumi2014phenomenology}. Starting instead from the theory of \emph{quantum processes} gives a version of the \emph{Quantum Integrated Information Theory} defined by Zanardi, Tomka and Venuti~\cite{zanardi2018quantum}, a major motivation for this work.

 Here we only outline the use of the categorical perspective for theories such as IIT. There is much scope for future work developing a richer study of integration and causality in monoidal categories, as well as for modifying IIT itself into a more categorical form, as discussed in~\cite{GeneralisedIITs}.  

 The article is structured as follows. After introducing process theories in Section \ref{sec:PTs} we use them to describe the key notions from IIT -- decompositions of objects (Section \ref{sec:Decompositions}), systems (Section \ref{sec:systems}) and cause and effect repertoires (Section \ref{sec:cause-and-effect}). We summarise how to define a generalised IIT from a process theory in Section \ref{sec:generalised-IITs} before giving examples in Section \ref{sec:examples} and discussing future work in Section \ref{sec:discussion}. The appendix contains some initial steps in developing a study of integration in monoidal categories.

\section{Process Theories} \label{sec:PTs}

We begin by introducing the framework of \emph{process theories} used throughout this work; for more detailed introductions we refer to \cite{coecke2010categories,coecke2017picturing}. The basic ingredients of such a theory are \emph{objects} and \emph{processes} between them. We depict a process from the object $A$ to the object $B$ as a box:
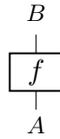
\[
\scalebox{1.0}{\begin{tikzpicture}
	\begin{pgfonlayer}{nodelayer}
		\node [style=none] (0) at (4, 1) {};
		\node [style=map] (1) at (4, -0) {$f$};
		\node [style=label] (2) at (4, -1.5) {$A$};
		\node [style=label] (3) at (4, 1.5) {$B$};
		\node [style=none] (4) at (4, -1) {};
	\end{pgfonlayer}
	\begin{pgfonlayer}{edgelayer}
		\draw [style=none] (4.center) to (0.center);
	\end{pgfonlayer}
\end{tikzpicture}}
\]
These processes may be \emph{composed} together to form new ones in several ways. Firstly, given a process such as $f$ above, and any other process $g$ from $B$ to $C$, we may compose them `in sequence'
to form a new one from $A$ to $C$, denoted:
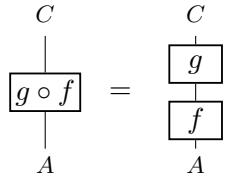
\[
\scalebox{1.0}{\input{./figures/composition-morphism.tikz}}
\]
Secondly, we may compose processes in parallel. Any two objects $A, B$ may be combined into a single object $A \otimes B$. Moreover any processes $f$ from $A$ to $B$, and $g$ from $C$ to $D$ may be placed `side-by-side' to form a new process:
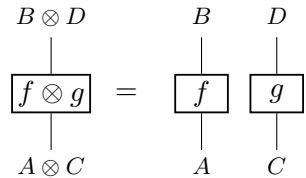
\[
\scalebox{1.0}{\input{./figures/tensor-morphism.tikz}}
\]
from $A \otimes C$ to $B \otimes D$. More generally, by combining these operations, many processes may all be plugged together to form more complex diagrams describing a single composite process.

As a convenience, any process theory is taken to come with the following. Firstly, any object $A$ come with an \emph{identity process}, depicted as a blank wire on $A$, which `does nothing' in that composing with it via $\circ$ leaves any process as it is. Secondly, it has a \emph{trivial object}, denoted $I$, which leaves objects alone when combining under $\otimes$. We depict $I$ as empty space:
\[
\scalebox{1.0}{\input{./figures/trivial-object-stuff.tikz}}
\]
Finally, we formally assume the presence of a special process $\tinyswap$ which allows us to `swap' any pair of wires over each other, along with a set of rules saying roughly that diagrams in the above sense are well-defined.

Mathematically, all of this is summarised by saying that a process theory is precisely a \emph{symmetric monoidal category} $(\catC, \otimes, I)$ with the processes as its \emph{morphisms}. Our diagrammatic rules correspond to the precise \emph{graphical calculus} for reasoning in such categories~\cite{selinger2011survey}.

We will often wish to refer to some special kinds of processes. Processes with `no input' in diagrams (and so formally with input object $I$) are called \emph{states}, and can be thought of as `preparations' of the physical system given by their output object:
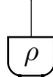
\[
\scalebox{1.0}{\begin{tikzpicture}
	\begin{pgfonlayer}{nodelayer}
		\node [style=point] (1) at (0, -0.25) {$\rho$};
		\node [style=none] (4) at (0, 1.25) {};
	\end{pgfonlayer}
	\begin{pgfonlayer}{edgelayer}
		\draw (4.center) to (1);
	\end{pgfonlayer}
\end{tikzpicture}
}
\]
Processes with no output, called \emph{effects}, may be thought of as `observations' we may record on our system. Finally, processes with neither input nor output are called \emph{scalars}.  It is common for theories to come with a \emph{probabilistic} interpretation meaning that each of their scalars $p$ correspond to a probability, or more generally an `unnormalised probability' $p \in \R^+$, with $r \otimes s = r \cdot s$ for scalars and the empty diagram given by $1$. In particular, the composition of a state with an effect
\[
\scalebox{1.0}{\begin{tikzpicture}
	\begin{pgfonlayer}{nodelayer}
		\node [style=point] (1) at (3, -0.25) {$\rho$};
		\node [style=copoint] (2) at (3, 1) {$e$};
		\node [style=none] (5) at (3, -0.25) {};
	\end{pgfonlayer}
	\begin{pgfonlayer}{edgelayer}
		\draw (2) to (5.center);
	\end{pgfonlayer}
\end{tikzpicture}
} \in \R^+
\]
corresponds to the `probability' of observing the effect $e$ in the state $\rho$. Such `generalised probabilistic theories' are a major focus of study in the foundations of physics~\cite{barrett2007information}. 

The theories we consider here will often come with further structure giving them a physical interpretation. Firstly, every object will come with a distinguished \emph{discarding} effect depicted
\[
\scalebox{1.0}{\begin{tikzpicture}
	\begin{pgfonlayer}{nodelayer}
		\node [style=none] (0) at (0, -0.75) {};
		\node [style=upground] (2) at (0, 0.25) {};
	\end{pgfonlayer}
	\begin{pgfonlayer}{edgelayer}
		\draw [style=none] (0.center) to (2);
	\end{pgfonlayer}
\end{tikzpicture}
}
\]
which we think of as the process of simply `throwing away' or `ignoring' a physical system. Similarly, every object should come with a distinguished \emph{completely mixed state} depicted as
\[
\scalebox{1.0}{\begin{tikzpicture}
	\begin{pgfonlayer}{nodelayer}
		\node [style=downground] (0) at (-3.5, -0.75) {};
		\node [style=none] (1) at (-3.5, 0.25) {};
	\end{pgfonlayer}
	\begin{pgfonlayer}{edgelayer}
		\draw [style=none] (1.center) to (0);
	\end{pgfonlayer}
\end{tikzpicture}
}
\]
which corresponds to preparing the object in a maximally `noisy' or `random' state. These processes should satisfy
\[
\scalebox{1.0}{\input{./figures/discard_axiom1i.tikz}}
\qquad
\scalebox{1.0}{\input{./figures/discardflip_axioms.tikz}}
\]
as well as 
\[
\scalebox{1.0}{\input{./figures/discard-I-extended.tikz}}
\qquad \qquad
\scalebox{1.0}{\input{./figures/discard-I-axiom.tikz}}
\]
for all objects $A, B$. We then define a process $f$ to be \emph{causal} when it satisfies
\[
\scalebox{1.0}{\input{./figures/causal.tikz}}
\]
or similarly as \emph{co-causal} if it preserves $\discardflip{}$. Discarding processes are in fact closely related to physical notions of causality; see for example~\cite{coecke2014terminality,chiribella2010probabilistic}. 

In such a probabilistic theory there is a unique process between any two objects, the \emph{zero process} $0$, such that composing any process via $\circ, \otimes$ with $0$ always yields $0$. 

At times we will assume our process theory also comes with a way of describing how similar any two causal states are. This amounts to a choice of \emph{distance function} on the set $\Stc(A)$ of causal states of each object $A$, providing a value
\[
d \left(
\scalebox{1.0}{\begin{tikzpicture}
	\begin{pgfonlayer}{nodelayer}
		\node [style=none] (0) at (1.25, 0.5) {};
		\node [style=point] (1) at (1.25, -0.5) {$a$};
		\node [style=label] (2) at (1.25, 1) {$A$};
	\end{pgfonlayer}
	\begin{pgfonlayer}{edgelayer}
		\draw (1) to (0.center);
	\end{pgfonlayer}
\end{tikzpicture}
}
\ , \ 
\scalebox{1.0}{\begin{tikzpicture}
	\begin{pgfonlayer}{nodelayer}
		\node [style=none] (0) at (1.25, 0.5) {};
		\node [style=point] (1) at (1.25, -0.5) {$b$};
		\node [style=label] (2) at (1.25, 1) {$A$};
	\end{pgfonlayer}
	\begin{pgfonlayer}{edgelayer}
		\draw (1) to (0.center);
	\end{pgfonlayer}
\end{tikzpicture}
}
\right) \in \mathbb{R}^+
\]
for each $a, b \in \Stc(A)$. Often this map $d$ will satisfy the axioms of a metric, but this is not required.


Our main examples of process theories will come with a notable extra feature, though this will not be necessary for our approach. In many theories it is possible to `reverse' any process, in that for any process $f$ there is another $f^\dagger$ in the opposite direction. We say a process theory has a \emph{dagger} when it comes with such a mapping 
\[
\scalebox{1.0}{\input{./figures/dagger-map.tikz}}
\]
which preserves composition and identity maps in an appropriate sense, and satisfies $f^{\dagger \dagger} = f$ for all $f$. The presence of a dagger is a common starting point in categorical approaches to quantum theory; see e.g. \cite{abramsky2004categorical,selinger2007dagger}.

Let us now meet our main examples of process theories with the above features.

\begin{example}[Classical probabilistic processes]
In the process theory $\Class$ of finite-dimensional probabilistic classical physics, the objects are finite sets $A, B, C, \dots$ and the processes $f$ from $A$ to $B$ are functions sending each element $a \in A$ to a `unnormalised probability distribution' over the elements of $B$, i.e~functions $f \colon A \times B \to \Rplus$. Composition of $f$ from $A$ to $B$ and $g$ from $B$ to $C$ is defined by
\[
(g \circ f)(a,c) = \sum_{b \in B} f(a,b) \cdot g(b,c)
\] 
In this theory the trivial object is the singleton set $I = \{\star\}$, with $\otimes$ given by the Cartesian product $A \times B$ and $(f \times g)(a,c)(b,d) = f(a,b) \cdot g(c,d)$. This theory is probabilistic, with scalars $r \in \Rplus$. 

Here $\discard{A}$ is the unique effect with $\discard{A}(a) = 1$ for all $a \in A$. A process $f$ is causal whenever it is stochastic, i.e.~sends each element $a \in A$ to a (normalised) probability distribution over the elements of $B$. Applying the process $\discard{}$ to some output wire of a process corresponds to \emph{marginalising} over the set which is discarded. 

States of an object are `$\Rplus$-distributions' over their elements, while causal states are normalised ones, i.e. probability distributions. The completely mixed state $\discardflip{A}$ is the uniform probability distribution, with $\discardflip{A}(a) = \frac{1}{|A|}$ for all $a \in A$. This theory also has a dagger by $f^\dagger(b,a) = f(a,b)$.

Rather than $\Class$ we will here work instead with the theory $\Classm$, defined in the same way, but with objects now being finite \emph{metric spaces} $(A,d)$. Each object $A$ now comes with a metric $d$ on its underlying set, with $A \otimes B = A \times B$ having the product metric. For each object $A$ we extend $d$ to a metric $d_W$ on probability distributions over $A$, i.e. causal states of $A$, called the \emph{Wasserstein metric} or \emph{Earth Mover's Distance} (EMD), definable e.g. by
\[
d_W(s,t) := \sup_f \{ \sum_{a \in A} f(a) \cdot s(a) - \sum_{a \in A} f(a)\cdot t(a) \}
\]
where the suprema is taken over all functions $f$ satisfying $| f(a) - f(b) | \leq d(a,b)$ for all $a, b$. $\Class$ itself may be given a metric on causal states in the same way by taking each object $A$ to have metric $d(a,b) = 1 - \delta_{a,b}$.
\end{example}

\begin{example}[Quantum Processes]
In the process theory $\Quant{}$ the objects are finite-dimensional complex Hilbert spaces $\hilbH, \hilbK, \dots$ and the processes from $\hilbH$ to $\hilbK$ are \emph{completely positive maps} $f \colon B(\hilbH) \to B(\hilbK)$ between their spaces of operators. Here $I= \mathbb{C}$ and $\otimes$ is the usual tensor product of Hilbert spaces and maps. States $\rho$ of an object $\hilbH$ may be identified with (unnormalised) \emph{density matrices}, i.e.~quantum states in the usual sense, as may effects. The effect $\discard{}$ sends each operator $a \in B(\hilbH)$ to its \emph{trace} $\Tr(a)$, and $\discardflip{}$ is the maximally mixed state on $\hilbH$, with density matrix $\frac{1}{\dim(\hilbH)} 1_{\hilbH}$. Here a process is causal precisely when it is trace-preserving, and the dagger is given by the Hermitian adjoint.
\end{example}

\begin{example}[Quantum-Classical Processes]
To combine $\Class$ and $\Quant{}$ we may use the theory $\FCStar$ whose objects are finite-dimensional \emph{C$^*$-algebras} $A, B, \dots$ and processes are completely positive maps $f \colon A \to B$, with $\otimes$ given by the standard tensor product, $I = \mathbb{C}$ and the dagger again by the Hermitian adjoint. Here $\discard{}$ sends each element $a \in A$ to its trace $\Tr(a) \in \mathbb{C}$, while $\discardflip{}$ corresponds to the rescaling $\frac{1}{d} 1$ of the element $1 \in A$, where $\Tr(1) = d$. Each C$^*$-algebra comes with a metric induced by its norm, providing a metric on states in the theory.

$\Class$ may be identified with the sub-theory of $\FCStar$ containing the commutative algebras, and $\Quant{}$ with those of the form $B(\hilbH)$ for some Hilbert space $\hilbH$. More general algebras are `quantum-classical', being given by direct sums of quantum algebras.
\end{example}

\section{Decompositions} \label{sec:Decompositions}

A central aspect of IIT is evaluating the level of integration of a process, and particularly of a state of some object. To do so we must compare the object in question against ways it may be \emph{decomposed}, as follows.

 Firstly, recall that a process $f$ from $A$ to $B$ is an \emph{isomorphism} when there is some (unique) $f^{-1}$ from $B$ to $A$ for which $f^{-1} \circ f$ and $f \circ f^{-1}$ are both identities. We write $A \simeq B$ when such an isomorphism exists.

\begin{definition}
In any process theory, a \emph{decomposition} of an object $S$ is a pair of objects $A, A'$ along with an isomorphism $S \simeq A \otimes A'$.

In a process theory with $\discard{}, \discardflip{}$ we will always consider decompositions whose isomorphisms are causal and co-causal.
\end{definition}

For short we often denote such a decomposition simply by $(A,A')$ and depict its isomorphism and inverse by 
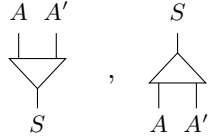
\[
\scalebox{1.0}{\input{./figures/decomp-new.tikz}} \ \ , \ \ 
\scalebox{1.0}{\input{./figures/decompflip-new.tikz}}
\]
respectively. The fact that they form an isomorphism means that 
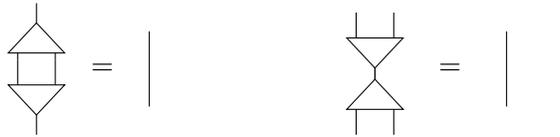
\[
\scalebox{1.0}{\input{./figures/decomp-conditions-nolabels.tikz}}
\]
One can go on to develop a general study of decompositions in process theories. Here we just note some of the basics, for more see Appendix~\ref{sec:appendix}.

Firstly, any decomposition has an induced \emph{complement} decomposition $(A,A')^\bot := (A',A)$, with isomorphism given by swapping its components:
 \[
\scalebox{1.0}{\input{./figures/decomp-swap.tikz}}
\]
All decompositions then satisfy $(A,A')^{\bot\bot}=(A,A')$. Moreover, any object always $S$ always comes with \emph{trivial decompositions} denoted $1 := (S,I)$ and $0 := (I,S)$ with $0 = 1^\bot$. Drawing either of their isomorphisms would just mean drawing a blank wire labelled by $S$. 

It is also useful to note when two decompositions of an object are `essentially the same'. We write $(A,A') \sim (B,B')$ and call both decompositions \emph{equivalent} when there exists isomorphisms $f, g$ with
\begin{equation} \label{eq:equiv-f-g}
\scalebox{1.0}{\input{./figures/decomp-equiv-new.tikz}}
\end{equation}
In a theory with $\discard{}, \discardflip{}$ we require moreover that $f, g$ are causal and co-causal.

We write $\Decomps(S)$ for the set of all equivalence classes of decompositions of $S$ under $\sim$. Often we  abuse notation and denote its members simply by $(A,A')$ instead of as equivalence classes $[(A,A')]_{\sim}$. It is easy to see that if two decompositions are equivalent then so are their complements, so that $(-)^\bot$ is well-defined on $\Decomps(S)$.

\begin{definition}
By a \emph{decomposition set} of an object $S$ in a process theory we mean a subset $\mathbb{D}$ of $\Decomps(S)$ containing $1$ and closed under $(-)^\bot$.
\end{definition}

Given any decomposition set $\mathbb{D}$ of $S$ and any $(A,A') \in \mathbb{D}$, we define the \emph{restriction} of $\mathbb{D}$ to $A$ via this decomposition to be the decomposition set
\[
\mathbb{D}|_A := 
\left\{
\scalebox{1.0}{\input{./figures/decABC.tikz}}
\large\mid 
\ \exists \scalebox{1.0}{\input{./figures/decACBdash.tikz}}
\text{ s.t. } \ 
\scalebox{1.0}{\input{./figures/decomp-large.tikz}}
\in \mathbb{D}
\right\}
\subseteq
\Decomps(A)
\]
Intuitively $\mathbb{D}|_A$ consists of all decompositions of $A$ which themselves can be extended to give a decomposition of $S$ belonging to $\mathbb{D}$, via $(A,A')$.

The most important examples of decomposition sets are the following.

\begin{example} \label{ex:simple-decomps}
Let $S$ be an object with a given isomorphism
\[
S \simeq S_1 \otimes \dots \otimes S_n
\]
representing $S$ as finite tensor of objects $S_i$ which we may call \emph{elements}. This induces a decomposition set $\mathbb{D}$ of $S$ whose elements correspond to subsets $J$ of the elements. For any such subset, defining $S_J := \bigotimes_{J} S_j$ we have a decomposition $S \simeq S_J \otimes S_{J'}$ where $J'$ is  the set of remaining elements. Then $\mathbb{D}|_{S_J}$ contains a decomposition for each $K \subseteq J$ in the same way.
\end{example}

Decomposition sets in terms of elements as above are the only kinds appearing in classical or quantum IIT. However more general ones would allow us to treat physical systems which are decomposable but not into any finite set of `elementary' subsystems.

\section{Systems} \label{sec:systems}

We now begin by seeing how each of the main components of IIT, or any `generalised IIT' in the sense of~\cite{GeneralisedIITs}, may be treated  starting from any given process theory $\catC$. The focus will be on a class of \emph{systems}, as follows.

\begin{definition}
By a \emph{system type} we mean a triple $\sysS = (S, \mathbb{D}, T)$ consisting of an object $S$ with a decomposition set $\mathbb{D}$ and a causal process
\[
\scalebox{1.0}{\input{./figures/time-evo.tikz}}
\]
which we call its \emph{time evolution}. A \emph{state} of $\sysS$ is simply a state of $S$ in $\catC$.
\end{definition}

The set $\mathbb{D}$ specifies the ways in which we will decompose our underlying system later when assessing integration. The process $T$ is intended to describe the way in which states of the system evolve over each single `time-step', via
\[
\scalebox{1.0}{\input{./figures/time-evo-states.tikz}}
\]
In what follows it will be useful to be able to restrict any state $s$ of our system to the components of any decomposition $(A,A') \in \mathbb{D}$ by setting 
\[
\scalebox{1.0}{\input{./figures/s-marginal-1.tikz}}
\]
and defining $s|_{A'}$ similarly.

A particular system of interest is the \emph{trivial system} $\sys{I}$ which has underlying object $I$, only a single decomposition $1=(I,I)=0$ in $\mathbb{D}$, and time evolution being the identity process. 

\subsection{Subsystems}

There are several operations one must carry out on systems in the contexts of IITs. The first is the taking of \emph{subsystems}. 

\begin{definition}
For each object $C$ belonging to some decomposition $(C,C') \in \mathbb{D}$, and each state $s$ of $\sysS$, the corresponding \emph{subsystem} of $\sysS$ is defined to be the system type $\sys{C}^s := (C, \mathbb{D}|_C, T|_C)$ with time evolution
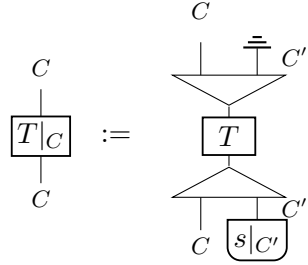
\[
\scalebox{1.0}{\input{./figures/restricted-evolution-full.tikz}}
\] 
\end{definition}

The above definition of the restricted evolution $T|_C$ comes from~\cite{oizumi2014phenomenology} and is intended to capture the evolution of a state of $C$ conditioned on the state of $C'$ being the restriction $\marg{s}{C'}$ of $s$.

\subsection{Cutting}

A second important operation involves removing (some or all) causal connections between the two different components of a decomposition of a system. For any system $\sysS=(S,\mathbb{D},T)$ and decomposition $(C,C') \in \mathbb{D}$, we should be able to form a new such \emph{cut} system of the form 
\[
\cut{\sysS}{(C,C')} = (S,\mathbb{D},\cut{T}{(C,C')})
\]
in which the new evolution $\cut{T}{(C,C')}$ should remove some influence between these two regions. 

The most straightforward form of cutting is a \emph{symmetric cut}, in which both components are fully disconnected from each other, with time evolution 
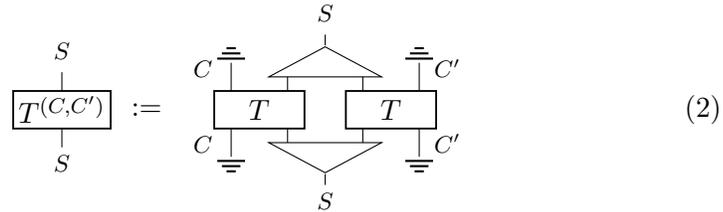
\begin{equation} \label{eq:cut-simple}
\scalebox{1.0}{\input{./figures/cut-evol-full.tikz}}
\end{equation}
(where the triangle above denotes $(C,C')^\bot$). However, other theories may use additional structure to carry out alternative notions of system cut, as we will see later. 

\section{Cause and Effect} \label{sec:cause-and-effect}

Central to any IIT is a notion of causal influence between any two possible subsystems of a system. These influences are captured in a pair of assignments called the \emph{cause repertoire} and \emph{effect repertoire} of the system. For our purposes it suffices to note that such cause and effect repertoires amount to specifying a pair of processes
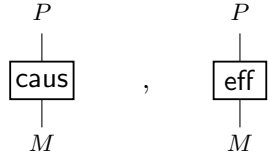
\[
\scalebox{1.0}{\input{./figures/cause-rep.tikz}} \qquad , \qquad \scalebox{1.0}{\input{./figures/eff-rep.tikz}}
\]
for each pair of underlying objects $M, P$ of subsystems $\sys{M}, \sys{P}$ of $\sys{S}$ via some state $s$. In this setting $M$ is typically called the `mechanism' and $P$ the `purview', and the above processes should capture the way in which the current state $m$ of $M$ constrains the previous or next state of $P$, respectively. These constraints are captured by the pair of states of $P$ given by plugging in the state $m$:
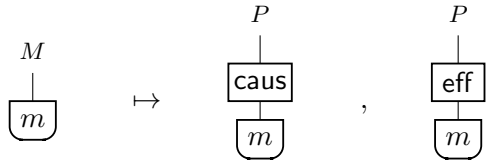
\[
\scalebox{1.0}{\begin{tikzpicture}
	\begin{pgfonlayer}{nodelayer}
		\node [style=none] (1) at (-0.25, 0.75) {};
		\node [style=none] (2) at (-0.25, -0.75) {};
		\node [style=label] (4) at (-0.25, 1.25) {$M$};
		\node [style=point] (5) at (-0.25, -0.5) {$m$};
	\end{pgfonlayer}
	\begin{pgfonlayer}{edgelayer}
		\draw (2.center) to (1.center);
	\end{pgfonlayer}
\end{tikzpicture}
}
\qquad
\mapsto 
\qquad
\scalebox{1.0}{\input{./figures/cause-rep-state.tikz}} \qquad , \qquad \scalebox{1.0}{\input{./figures/eff-rep-state.tikz}}
\]
We will additionally require the processes $\caus, \eff$ to be \emph{weakly causal} in the sense that whenever the state $m$ is causal then each of the above states must either be causal or $0$.

\begin{example} \label{ex:naive-caus-reps}
For any process theory there is a simple choice of effect repertoire, given by
\begin{equation} \label{eq:eff-rep-naive}
\scalebox{1.0}{\input{./figures/e-rep-naive-decomps.tikz}}
\end{equation}
where $T$ is the time evolution of the system. If our process theory has a dagger there is a similar straightforward choice of cause repertoire given by 
\begin{equation} \label{eq:caus-rep-naive}
\scalebox{1.0}{\input{./figures/c-rep-naive-decomps.tikz}}
\end{equation}
though this may not be weakly causal in our above sense if $T^\dagger$ is not causal. In a probabilistic process theory we should instead have that  
\begin{equation} \label{eq:caus-rep-naive-normalised}
\scalebox{1.0}{\input{./figures/c-rep-naive-decomps-normalised.tikz}}
\end{equation}
where $\lambda_m$ is the unique \emph{normalisation} scalar for the right-hand state, making it causal if it is non-zero (and being zero otherwise). It is not in general possible to define a process $\caus$ in terms of its action on states $m$ in this way, but this is possible for example in $\Class$, $\Quant{}$ or $\FCStar$.
\end{example}

However the cause or effect repertoires are specified, we will need to compare their values in a fixed state while varying $P$. To do so, for each state $s$ of $\sys{S}$ and each such $M, P$ we define the \emph{cause repertoire at $s$} to be the state of $S$ given by 
\begin{equation} \label{eq:caus-rep-value}
\scalebox{1.0}{\input{./figures/extended-cause-value-full.tikz}}
\end{equation}
The features of this diagram have special names in \cite{oizumi2014phenomenology}; the right-hand $\caus$ state above, given by taking mechanism $M = I$, is called the \emph{unconstrained} cause repertoire, and the whole process above $s|_M$ in the diagram is called the \emph{extended cause repertoire} at $M, P$. Defining them in this way allows us to compare the repertoire values for varying $M, P$.

Similarly, $\eff_s(M,P)$, the \emph{effect repertoire at $s$}, and the \emph{unconstrained} and \emph{extended effect repertoire} are all defined in terms of $\eff$ in the same way.

\subsection{Decomposing repertoires} \label{subsec:decomp-rep}

In an IIT we must assess how integrated each of these repertoire values are at a given state
. This involves comparing the repertoires with how they behave under decomposing each of $M$ and $P$. For any decompositions $(M_1, M2) \in \mathbb{D}|_M$ of $M$ and $(P_1, P_2) \in \mathbb{D}|_P$ of $P$, the \emph{decomposed} cause repertoire process is defined by
\begin{equation} \label{eq:decomp-caus}
\scalebox{1.0}{\input{./figures/caus-decompose.tikz}}
\end{equation}
We then define the state $\caus^{P_1,P_2}_{s,M_1,M_2}(M,P)$ just like \eqref{eq:caus-rep-value} but replacing $\caus$ with the process \eqref{eq:decomp-caus}. We decompose the effect repertoire in just the same way in terms of $\eff$.

\section{Generalised IITs} \label{sec:generalised-IITs}

In summary, let $\catC$ be a process theory coming with the features $\discard{}, \discardflip{}, d$ of Section \ref{sec:PTs}. To define an integrated information theory we must specify:
\begin{enumerate}
\item \label{enum:sys}
	a class $\Sys$ of system types, containing $\sys{I}$ and closed under taking subsystems;

\item \label{enum:cuts}
    a definition of system cuts, under which $\Sys$ is closed;
\item \label{enum:CE}
    a choice of weakly causal processes $\caus, \eff$ between the underlying objects $M, P$ of each pair of subsystems $\sys{M}, \sys{P}$ via some state $s$, of any system $\sys{S}$.
\end{enumerate}

More precisely, this provides the \emph{data} of a generalised integrated information theory in the sense of~\cite{GeneralisedIITs}.  From this data we may now use the \emph{IIT algorithm} from~\cite{oizumi2014phenomenology} to calculate the usual objects of interest in IIT.

\subsection{The IIT Algorithm}
We now briefly summarise this algorithm as treated in the general setting in~\cite{GeneralisedIITs}, to which we refer for more details. Let us fix a `current' state $s$ of a system $\sys{S}$. 
Firstly, the level of \emph{integration} of each value of the cause repertoire is defined by 
\begin{equation} \label{eq:integration-calc}
\phi(\caus_s(M,P))
 \ := \ 
\min d(\caus_s(M,P) \ , \ \caus^{P_1,P_2}_{s,M_1,M_2}(M,P))
\end{equation}
where the minima is taken over all pairs of decompositions of $M, P$ which are not both trivial, i.e.~equal to $1$. \footnote{When $\caus_s(M,P) = 0$ we alternatively set $\phi = 0$.} The integration level $\phi(\eff_s(M,P))$ is defined similarly in terms of $\eff$. 

For each choice of mechanism $M$, its \emph{core cause} $P^c$ and \emph{core effect} $P^e$ are the purviews $P$ with maximal $\phi$ values for $\caus, \eff$ respectively. The minima of their corresponding $\phi$ values is then denoted by $\phi(M)$. We then associate to $M$ and object called its \emph{concept} $\concept{M}$, essentially defined as the triple
\[
(\caus_s(M,P^c),\eff_s(M,P^e), \phi(M))
\]
More precisely, in \cite{GeneralisedIITs}, $\concept{M}$ is given by the pair of above repertoire values with each `rescaled' by $\phi(M)$. 

The tuple $\QShape(s)$ of all these concepts, for varying $M$, is called the \emph{Q-shape} $\QShape(s)$ of the state $s$. The collection of all possible such tuples is denoted by $\Exp(\sys{S})$. The level of integration of $\QShape(s)$ is calculated similarly to~\eqref{eq:integration-calc} by considering all possible cuts of the system. The subsystem $\sys{M}$ of $\sys{S}$ whose Q-shape is itself found to be most integrated is called the \emph{major complex}. Rescaling this Q-shape $\QShape(\sys{M},s|_M)$ according to its level of integration, and using an embedding $\Exp(\sys{M}) \hookrightarrow \Exp(\sys{S})$ we finally obtain a new element $\Exp(s) \in \Exp(\sysS)$.

The claim of a generalised IIT with regards to consciousness is that $\Exp(\sysS)$ is the space of all possible conscious experiences of the system $\sysS$, and that $\Exp(s)$ is the particular experience attained when it is in the state $s$, with intensity $\Phi(s) := \norm \Exp(s) \norm$.

\begin{Remark}
Let us make explicit how the specification of \ref{enum:sys}, \ref{enum:cuts}, \ref{enum:CE} above provides the data of an IIT in the sense of~\cite{GeneralisedIITs}. 
The system class of the theory is $\Sys$, and $\caus_s(M,P), \eff_s(M,P)$ and their decompositions are as outlined in Section \ref{subsec:decomp-rep}. When $\catC$ is probabilistic and has distances $d(a,b)$ defined for \emph{arbitrary} states $a, b$ of an object $A$, we may define the space of \emph{proto-experiences} $\PExp(\sys{S})$ of a system $\sys{S}$ to be simply its set of states, with
\[
\left \| \scalebox{1.0}{\begin{tikzpicture}
	\begin{pgfonlayer}{nodelayer}
		\node [style=point] (5) at (-4, -0.25) {$s$};
		\node [style=none] (6) at (-4, 0.75) {};
	\end{pgfonlayer}
	\begin{pgfonlayer}{edgelayer}
		\draw (6.center) to (5);
	\end{pgfonlayer}
\end{tikzpicture}
} \! \right \| := \scalebox{1.0}{\begin{tikzpicture}
	\begin{pgfonlayer}{nodelayer}
		\node [style=point] (5) at (-4, -0.25) {$s$};
		\node [style=none] (6) at (-4, 0.75) {};
		\node [style=upground] (7) at (-4, 1) {};
	\end{pgfonlayer}
	\begin{pgfonlayer}{edgelayer}
		\draw (6.center) to (5);
	\end{pgfonlayer}
\end{tikzpicture}
}
\]
However, if $d$ is only defined on causal states, as in classical IIT, to follow the algorithm from~\cite{GeneralisedIITs} one must instead set $\PExp(\sys{S}) := \Stc(S) \times \mathbb{R}^+$ as in~\cite[Ex. 3]{GeneralisedIITs}. For either choice, for any subsystem $\sys{M}$ of $\sys{S}$ we obtain an embedding $\PExp(\sys{M}) \hookrightarrow \PExp(\sys{S})$ by composing alongside $\discardflip{M^\bot}$, and this can be seen to provide a further embedding $\Exp(\sys{M}) \hookrightarrow \Exp(\sys{S})$. 
\end{Remark}


\section{Examples} \label{sec:examples}

Let us now meet several examples of IITs defined from process theories.

\subsection{Generic IITs}

Let $\catC$ be any operational process theory coming with a dagger on processes. 
We define a generalised IIT denoted $\BasicIIT(\catC)$ by taking as systems all tuples $\sys{S} = (S,\mathbb{D}, T)$ of an object $S$ in $\catC$ along with a causal process $T$ and a decomposition set $\mathbb{D}$ induced by a single isomorphism $S \simeq \bigotimes^n_{i=1} S_i$ in terms of elements $S_i$, as in Example~\ref{ex:simple-decomps}. As before each partition of these elements gives a decomposition of $S$. We define system cuts to be symmetric as in~\eqref{eq:cut-simple} and the repertoires are defined in the straightforward sense of \eqref{eq:eff-rep-naive}, \eqref{eq:caus-rep-naive}.

\begin{Remark}
We can extend this example in to ways. Firstly we may allow systems $\sys{S}$ to come with arbitrary finite decomposition sets $\mathbb{D}$ of $S$. Secondly, we may extend the definition to theories without daggers by instead simply requiring each system $\sys{S}$ to come with a process $T^-$ describing `reversed time evolution', and then define the cause repertoire by replacing $T^\dagger$ with $T^-$. 
\end{Remark}

\subsection{Classical IIT}

The `classical' IIT version 3.0 of Tononi and collaborators~\cite{oizumi2014phenomenology} is built on the process theory $\Classm$. 
As such a toy model of the theory is provided by $\BasicIIT(\Class)$. However IIT 3.0 itself differs from this theory, using some more specific features of the process theories $\Class$ and $\Classm$ which we now describe.  

Firstly, note that in these classical process theories, for each object $A$, each element $a \in A$ corresponds to a unique state given by the point distribution at $a$, as well as a unique effect, namely the map sending $a$ to $1$ and all other elements of $A$ to $0$. We denote this state and effect both simply by $a$.\footnote{Typically these are the only kinds of `state' considered, e.g. in \cite{oizumi2014phenomenology} and even in our related article \cite{GeneralisedIITs}. In contrast here the term `state' would include all distributions over $A$, i.e. all states of the process theory $\Class$.}

Any process $f$ from $A$ to $B$ is determined entirely by its values on these special states and effects since we have   
\[
f(a,b) = \scalebox{1.0}{\input{./figures/map-on-values.tikz}}
\]
for all $a \in A, b \in B$.

Another special feature of these classical process theories is that each object $A$ comes with a distinguished \emph{copying} process from $A$ to $A \otimes \dots \otimes A$, for any number of copies of $A$, as well as a \emph{comparison} process in the opposite direction. We denote and define these respectively by the rules 
\[
\scalebox{1.0}{\input{./figures/copymaps-eq-a.tikz}}
\qquad \qquad
\scalebox{1.0}{\input{./figures/comparisonmaps-eq-a.tikz}}
\]
for all $a \in A$. Abstractly, these operations form a canonical commutative \emph{Frobenius algebra} on each object, and there is no such canonical algebra on each object in $\Quant{}$ due to the \emph{no-cloning} theorem \cite{coecke2013new}. 

We may now describe IIT 3.0 itself as follows. 

\subsubsection{Systems}

In this theory systems are defined similarly to $\BasicIIT(\Class)$, being given by a set $S$ given as a product $S \simeq \bigotimes^n_{i=1} S_i$ of `elements' $S_i$, along with a causal (i.e.~stochastic) evolution $T$ on $S$. Additionally in \cite{oizumi2014phenomenology} each evolution $T$ is required to satisfy the property of \emph{conditional independence}, which states that for all $s, t \in S$, with $t=(t_1,\dots, t_n)$ for some $t_i \in S_i$ we have 
\[
\scalebox{1.0}{\input{./figures/con-indep.tikz}}
\]
where for each element $S_i$ we define the process $T_i$ by
\[
\scalebox{1.0}{\input{./figures/Ti.tikz}}
\]
having depicted the isomorphism $S \simeq \bigotimes^n_{i=1} S_i$ by the triangle above. In other words, conditional independence states that the probabilities for the next state of each element $S_i$ are independent. Equivalently, $T$ must satisfy
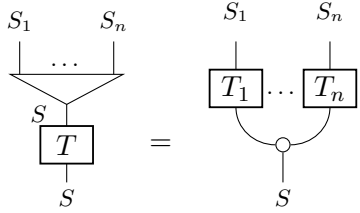
\[
\scalebox{1.0}{\input{./figures/cond-indep-2.tikz}}
\]

\subsubsection{Cuts}
Rather than our earlier symmetric cuts, the system cuts used in IIT 3.0 are \emph{directional}. For any decomposition $(C,C')$ of $S$ with $C = \bigotimes_{j \in J} S_j$ for some subset of notes indexed by $J \subseteq \{ 1, \dots ,n \}$, we define the cut evolution $\cut{T}{(C,C')}$ using conditional independence by setting
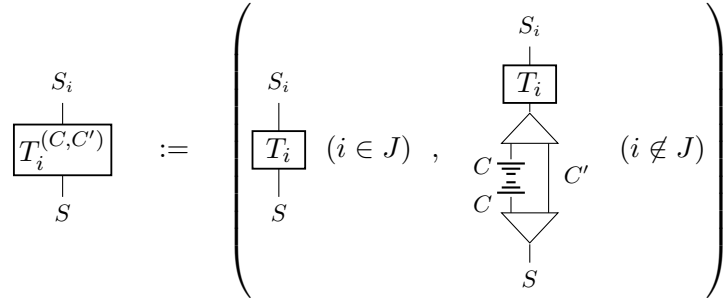
\[
\scalebox{1.0}{\input{./figures/cut-T-1.tikz}}
\quad
:= 
\quad
\left(
\scalebox{1.0}{\input{./figures/cut-T-2.tikz}}
 \ (i \in J)
\ \
 ,
 \ \
\scalebox{1.0}{\input{./figures/cut-T-3.tikz}}
 \ (i \not \in J)
\right)
\]
In other words, in the cut system all causal connections $C \to C'$ are replaced by noise, while all those into $C$ remain intact. 

\subsubsection{Repertoires}

Let us now define the processes $\caus, \eff$ between a pair of objects $M$ and $P$, with $M = \bigotimes_{i=1}^k M_i$ and $P = \bigotimes_{j=1}^r P_j$ for some subsets $ \{M_1, \dots, M_k\}$ and $\{P_1, \dots, P_r\}$ of elements of the system. 

We begin with $\eff$. When $P$ is simply a single element $P_j$, $\eff$ is defined exactly as in \eqref{eq:eff-rep-naive}. For more general $P$ we define $\eff$ to again satisfy a form of conditional independence, so that
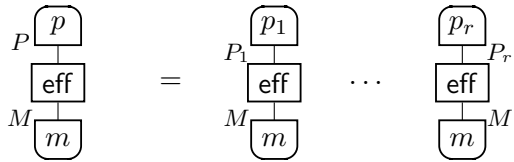
\[
\scalebox{1.0}{\input{./figures/erep-prod-simple-2.tikz}}
\]
for all $m \in M, p=(p_1,\dots,p_r) \in P$. Equivalently, we have that 
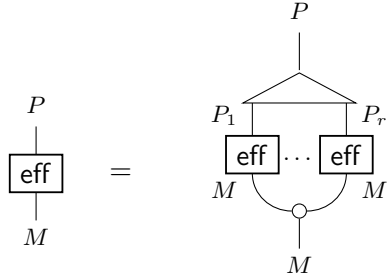
\[
\scalebox{1.0}{\input{./figures/erep-prod-new-2.tikz}}
\]
In a similar fashion, whenever $M$ is a single element $M_i$ we define $\caus$ from $M$ to $P$ as in \eqref{eq:caus-rep-naive-normalised}, while for more general $M$ we require that
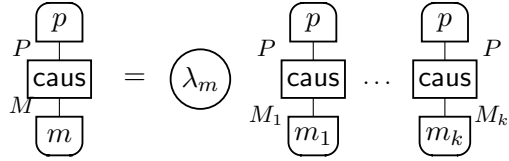
\[
\scalebox{1.0}{\input{./figures/crep-prod-3-scalar.tikz}}
\]
for all $m =(m_1, \dots, m_k) \in M$ and $p \in P$, where $\lambda_m$ is the normalisation scalar making $\caus \circ m$ a causal state (probability distribution) if it is non-zero, or $\lambda_m = 0$ otherwise. Equivalently, this means that 
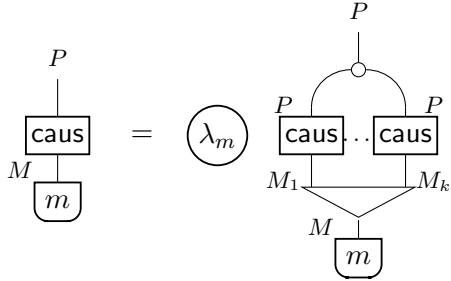
\[
\scalebox{1.0}{\input{./figures/crep-prod-new-2-scalar.tikz}}
\]
for each $m \in M$. This concludes the data of classical IIT.

\subsection{Quantum IIT}

Zanardi, Tomka and Venuti have proposed a quantum extension of classical IIT~\cite{zanardi2018quantum}. In fact it is comparatively much simpler to describe in our approach, being precisely the theory $\BasicIIT(\Quant{})$.

Explicitly, systems in this theory are given by finite-dimensional complex Hilbert spaces $\hilbH$ along with a given decomposition into elements $\hilbH \simeq \bigotimes^n_{i=1} \hilbH_i$ and a completely positive trace-preserving map $T$ on $B(\hilbH)$. States and repertoire values are given by density matrices $\rho$. In this theory each Q-shape $\QShape(\rho)$ may be encoded as a single positive semi-definite operator on the space $(\mathbb{C}^2)^{\otimes n} \otimes \mathbb{C}^2 \otimes \hilbH$, as discussed in~\cite{zanardi2018quantum}.

\subsection{Quantum-Classical IIT}
We may now define a version of \emph{quantum-classical IIT} as $\BasicIIT(\FCStar)$. This synthesizes quantum IIT with the toy version $\BasicIIT(\Class)$ of classical IIT, containing both kinds of systems. In future it would be desirable to synthesise quantum IIT with IIT 3.0 proper. Since the latter relies on the presence of copying maps, this may be achievable using the more general notion of a \emph{leak} on a C$^*$-algebra \cite{selby2017leaks}.


\section{Outlook} \label{sec:discussion}

In this article we have simply aimed to show how integrated information theory, and its generalisations to other domains of physics, may be studied categorically. There are many avenues for future work.

Firstly, we have so far made no requirements on the cause and effect repertoire processes $\caus$, $\eff$. To be fit for their name these processes should be required to satisfy axioms which ensure they have a causal interpretation, ideally determining them uniquely within any given process theory. Monoidal categories provide a natural setting for the study of causality, a major contemporary topic in the foundations of physics \cite{kissinger2017categorical}.

At a higher level, it seems natural for the class of systems $\Sys$ of a generalised IIT to itself form a category. The theory itself should then give a functor into another category $\Expcat$ of (spaces of) phenomenal experiences; a formalization of the latter is for example given in~\cite{GeneralisedIITs}. 

Making IIT functorial in this way will likely involve modifying it to be more natural from a categorical perspective. Indeed the IIT algorithm as currently stated is not even well-defined, for example relying on the unique existence of core purviews which are not guaranteed. Developing a useful notion of integration applicable to any monoidal category may help to resolve such problems. We make some first steps in this direction in the appendix.

\medskip

\Thanks {\small{\textsc{Acknowledgements:}}
We would like to thank the organizers and participants of the \emph{Workshop on Information Theory and Consciousness}
at the Centre for Mathematical Sciences of the University of Cambridge, of the \emph{Modelling Consciousness Workshop} in Dorfgastein
and of the \emph{Models of Consciousness Conference} at the Mathematical Institute of the University of Oxford for discussions on this topic.
Much of this work was carried out while Sean Tull was under the support of an EPSRC Doctoral Prize at the University of Oxford, from November 2018 to July 2019, and while Johannes Kleiner was under the support of postdoctoral funding at the Institute for Theoretical Physics of the Leibniz University of Hanover. We would like to thank both institutions.
}

\bibliographystyle{alpha}
\bibliography{iit.bib}


\appendix
\section{Decompositions and Integration} \label{sec:appendix}

Here we briefly mention a few further results about decompositions of objects in process theories; we leave a detailed study of their properties to future work. 

 Our earlier definition of $\mathbb{D}|_A$ was based on an idea of one decomposition as being `contained in' another. Let us make this precise.

\begin{definition}
Let $S$ be an object in a process theory and $(A,A')$, $(B,B')$ two decompositions. We write that $(A,A') \preceq (B,B')$ whenever there exists an object $C$ and decompositions $(A,C)$ of $B$ and $(B', C)$ of $A'$ such that 
\begin{equation} \label{eq:pre-order-new}
\scalebox{1.0}{\input{./figures/decomp-preorder-new.tikz}}
\end{equation}
\end{definition}
Intuitively, this states that $A$ is contained in $B$ (as is $B'$ within $A'$) in a way compatible with these decompositions.

\begin{Lemma} \label{lem:decomp-lemma}
Let $S$ be an object in a process theory. Then $\preceq$ forms a pre-order on the set of decompositions of $S$, with top element $1$ and bottom element $0$, and $(-)^\bot$as an involution.
\end{Lemma}
\begin{proof}
We always have $(A,A') \preceq (A,A')$ by taking $C=I$ and using the decompositions $1$ and $0$ on $A$ in~\eqref{eq:pre-order-new}. Similarly $(A,A') \preceq 1$ by taking $C=A'$. To see that $(-)^\bot$ is an involution, suppose that $(A,A') \preceq (B,B')$ as above. Then we have $(B,B')^\bot \preceq (A,A')^\bot$ since
\[
\scalebox{1.0}{\input{./figures/decomp-preorder-swaps2.tikz}}
\]
 Hence we always have $0 = 1^\bot \preceq (A,A')$ for all $(A, A')$. 
For transitivity, note that whenever $(A,A') \preceq (B,B') \preceq (C,C')$ via some respective objects $D, E$ then we have
 \[
\scalebox{1.0}{\input{./figures/decomp-trans-new2.tikz}}
 \]
 so that $(A,A') \preceq (C,C')$ via the above decompositions $(D \otimes E, C')$ of $A'$ and $(A, D \otimes E)$ of $C$.

\end{proof} 

Recall that in any category, a \emph{sub-object} of an object $A$ is an (isomorphism class of a) monomorphism $m \colon M \to A$. It is \emph{split} when $e \circ m = \id{M}$ for some $e$. The sub-objects of $A$ form a partial order $\Sub(A)$.

\begin{Lemma}
In any process theory with $\discard{}, \discardflip{}$, for any object $S$:
\begin{enumerate}
\item \label{enum:sub-ob}
Any decomposition $(A,A')$ of $S$ makes $A$ a split sub-object of $S$ via
\begin{equation} \label{eq:split-sub}
\scalebox{1.0}{\input{./figures/split-sub-2.tikz}}
 \ \ , \ \ 
\scalebox{1.0}{\input{./figures/split-sub-1.tikz}}
\end{equation}
Moreover if  $(A,A') \preceq (B,B')$ then $A \leq B$ in $\Sub(S)$.
\item \label{enum:partial-order}
 $\preceq$ restricts to a partial order $\leq$ on $\Decomps(S)$, again with top element $1$, bottom $0$ and involution $(-)^\bot$. 
\end{enumerate}
\end{Lemma}
\begin{proof}
\ref{enum:sub-ob}: We have 
\[
\scalebox{1.0}{\input{./figures/split-sub-arg.tikz}}
\]
If $(A,A') \preceq (B,B')$ then the splitting for $A$ factors over that for $B$ since:
\[
\scalebox{1.0}{\input{./figures/split-sub-contained.tikz}}
\]
It follows that $A \leq B$ in $\Sub(S)$.

\ref{enum:partial-order}:
We need to show that any two decompositions $(A,A')$ and $(B,B')$ are equivalent under $\preceq$ precisely when they are equivalent in the sense of~\eqref{eq:equiv-f-g}. Firstly, if there exists causal and co-causal isomorphisms $f, g$ making \eqref{eq:equiv-f-g} hold,  then we have 
\[
\scalebox{1.0}{\input{./figures/decomp-transfer.tikz}}
\]
Viewing $f^{-1}$ and $g$ as decompositions $(A, I)$ of $B$ and $(I, B')$ of $A'$, respectively, this gives that $(B,B') \preceq (A,A')$. Then $(A, A') \preceq (B,B')$ holds similarly. 

Conversely, if $(A, A') \preceq (B,B') \preceq (A,A')$, via respective objects $C, D$ then
\[
\scalebox{1.0}{\input{./figures/decomp-more-proof.tikz}}
\]
Since the right-hand map is an epimorphism by the first part, this gives that
\[
\scalebox{1.0}{\input{./figures/decomp-more-proof-2.tikz}}
\]
Dually, composing in the other order gives the identity on $A$, making these causal and co-causal isomorphisms $A \simeq B$. Similarly we obtain such isomorphisms $A' \simeq B$'. Then we have
\[
\scalebox{1.0}{\input{./figures/decomp-equiv-proof.tikz}}
\]
as required.
Now \ref{enum:partial-order} follows since any pre-order restricts to a partial order on its set of equivalence classes, and so $\preceq$ becomes a partial order $\leq$ on $\Decomps(S)$. It is easy to see that the earlier properties of $1,0,(-)^\bot$ carry over to $\leq$. 
\end{proof}

\subsection{Integration}

Let us briefly allude to how integration may generally be studied and quantified using decomposition sets.

Suppose we have objects $S, S'$ with given decomposition sets $\mathbb{D}, \mathbb{D}'$ and for each $(A,A') \in \mathbb{D}$ and $(B,B') \in \mathbb{D}'$ a process $f^{B}_A$ from $A$ to $B$. We denote $f^{S'}_S$ simply by $f$. Whenever we have a given distance function $d$ on the set of processes from $S$ to $S'$,  we may define the level of \emph{integration} of the family $(f^B_A)_{A,B}$ as 
\[
\phi(f)
:=
\min_{\mathbb{D} \times \mathbb{D}'}
d\left(
\scalebox{1.0}{\begin{tikzpicture}
	\begin{pgfonlayer}{nodelayer}
		\node [style=none] (1) at (12, -1) {};
		\node [style=none] (3) at (12, 1) {};
		\node [style=label] (6) at (12, -1.5) {$S$};
		\node [style=map] (12) at (12, 0) {$f$};
		\node [style=label] (20) at (12, 1.5) {$S'$};
	\end{pgfonlayer}
	\begin{pgfonlayer}{edgelayer}
		\draw (3.center) to (1.center);
	\end{pgfonlayer}
\end{tikzpicture}
}
,
\scalebox{1.0}{\input{./figures/f-w-decomps.tikz}}
\right)
\]
where we exclude the top element $(1,1)$ of $\mathbb{D} \times \mathbb{D}'$ in the minimisation.

\begin{example}
Given any process $f$ from $S$ to $S'$ we may define such a family $(f^B_A)_{A,B}$ with $f^{S'}_S = f$ by setting
\[
\scalebox{1.0}{\input{./figures/f-restr-1.tikz}}
\]
\end{example}

\begin{example}
Our earlier description of the IIT algorithm precisely includes evaluating the integration level of each of the families of processes $(\caus)_{M,P}$ and $(\eff)_{M,P}$ using the state-dependent distance
\[
d_m \left(
\scalebox{1.0}{\input{./figures/fMP.tikz}},
\scalebox{1.0}{\input{./figures/gMP.tikz}}
\right)
:=
d \left(
\scalebox{1.0}{\input{./figures/f_m.tikz}},
\scalebox{1.0}{\input{./figures/g_m.tikz}}
\right)
\]
where $m = s|_M$ and $d$ is the distance on $\St(S)$.
\end{example}



\end{document}